\documentclass[journal]{IEEEtran}

\usepackage{amsfonts}
\usepackage{amsmath}
\usepackage{amssymb}

\newtheorem{theorem}{\indent Theorem}[section]
\newtheorem{lemma}[theorem]{\indent Lemma}
\newtheorem{proposition}[theorem]{\indent Proposition}
\newtheorem{corollary}[theorem]{\indent Corollary}
\newtheorem{EXAMPLE}{\indent Example}[section]

\newcommand{\graph}{{\mathcal{G}}}
\newcommand{\code}{{\mathcal{C}}}
\newcommand{\Code}{{\mathsf{C}}}
\newcommand{\Edge}{{\mathfrak{E}}}

\newcommand{\GF}{{\mathrm{GF}}}

\newcommand{\ff}{\mathbb{F}}

\newcommand{\ffq}{\mathbb{F}_q}
\newcommand{\half}{{\textstyle\frac{1}{2}}}
\newcommand{\thalf}{{\textstyle\frac{3}{2}}}
\newcommand{\fhalf}{{\textstyle\frac{5}{2}}}

\newcommand{\cH}{{\mathcal{H}}}

\newcommand{\cE}{{\mathcal{E}}}

\newcommand{\cP}{{\mathcal{P}}}
\newcommand{\cQ}{{\mathcal{Q}}}

\newcommand{\alp}{{\mathsf{a}}} 
\newcommand{\bet}{{\mathsf{b}}}

\newcommand{\dist}{{\mathsf{d}}}
\newcommand{\distance}{{\mathsf{d}}}

\newcommand\nn{{\mathbb N}}

\newcommand{\blda}{{\mbox{\boldmath $a$}}}
\newcommand{\bldaa}{{\mbox{\scriptsize \boldmath $a$}}}
\newcommand{\bldb}{{\mbox{\boldmath $b$}}}
\newcommand{\bldbb}{{\mbox{\scriptsize \boldmath $b$}}}
\newcommand{\bldc}{{\mbox{\boldmath $c$}}}

\newcommand{\bldd}{{\mbox{\boldmath $d$}}}
\newcommand{\blddd}{{\mbox{\scriptsize \boldmath $d$}}}

\newcommand{\bldf}{{\mbox{\boldmath $f$}}}

\newcommand{\bldgamma}{{\mbox{\boldmath $\gamma$}}}

\newcommand{\bldh}{{\mbox{\boldmath $h$}}}

\newcommand{\blds}{{\mbox{\boldmath $s$}}}

\newcommand{\bldtau}{{\mbox{\boldmath $\tau$}}}

\newcommand{\bldw}{{\mbox{\boldmath $w$}}}

\newcommand{\bldx}{{\mbox{\boldmath $x$}}}

\newcommand{\bldy}{{\mbox{\boldmath $y$}}}

\newcommand{\bldz}{{\mbox{\boldmath $z$}}}
\newcommand{\bldxi}{{\mbox{\boldmath $\xi$}}}
\newcommand{\bldXi}{{\mbox{\boldmath $\Xi$}}}

\newcommand{\alphaf}{{\mu_B}}
\newcommand{\betaf}{{\mu_A}}

\newcommand{\entropy}{\mathsf{h}}

\newcommand{\quarter}{{\textstyle\frac{1}{4}}}

\newcommand{\ost}{{\textstyle\frac{1}{16}}}

    \def\squarebox#1{\hbox to #1{\hfill\vbox to #1{\vfill}}}

\renewcommand{\baselinestretch}{1.0}

\title{Correcting a Fraction of Errors in Nonbinary Expander Codes with Linear Programming%
    \thanks{
    This work was supported in part by the Claude Shannon Institute for
    Discrete Mathematics, Coding and Cryptography (Science Foundation Ireland Grant 06/MI/006), and in
    part by the National Research Foundation of Singapore (Research Grant NRF-CRP2-2007-03).
    Part of this work was presented at \emph{IEEE International Symposium on Information Theory 2009}, Seoul, Korea. 
    \newline
    V. Skachek was with the Claude Shannon Institute and the School of Mathematical Sciences, 
    University College Dublin, Belfield, Dublin 4, Ireland. 
    He is now with the Division of Mathematical Sciences, School of Physical and Mathematical Sciences,
    Nanyang Technological University, 21 Nanyang Link, Singapore 637371, 
    e-mail: {\tt Vitaly.Skachek@ntu.edu.sg}. 
    }
   }

\author{Vitaly Skachek,~\IEEEmembership{Member,~IEEE}}

\begin{document}

\maketitle

\begin{abstract}
A linear-programming decoder for \emph{nonbinary} 
expander codes is presented. It is shown that 
the proposed decoder has the nearest-neighbor 
certificate properties. It is also shown that 
this decoder corrects any pattern of errors 
of a relative weight up to approximately $\quarter \delta_A \delta_B$
(where $\delta_A$ and $\delta_B$ are the relative 
minimum distances of the constituent codes). 
\end{abstract}


\section{Introduction}
\label{sec:introduction}

Low-density parity check (LDPC) codes have become very popular in recent years due to their excellent performance
under message-passing (MP) decoders. Yet, our understanding of LDPC codes and their decoders is still limited. 
While most of the research to date was devoted to binary LDPC codes, there are works suggesting 
that nonbinary LDPC codes combined with high-order modulation schemes can possibly 
outperform their binary counterparts (at a price of higher decoding complexity)~\cite{LSLG},~\cite{Deepak}. 

For a binary case, a new approach toward understanding of LDPC codes was suggested in~\cite{Feldman-thesis} and~\cite{Feldman}: it was proposed to decode binary LDPC codes using linear-programming (LP) decoder, and important connections between the linear-programming decoding and the message-passing decoding were established (see also~\cite{KV-IEEE-IT},~\cite{Wiberg}). 
In particular, it was shown that the events of LP decoder failures 
are caused by so-called \emph{pseudocodewords}, 
and those pseudocodewords are, in turn, related to the failure events 
of the message-passing decoders.  

These results were generalized in~\cite{FSBG},~\cite{FSBG-2} toward nonbinary LDPC codes and coded modulations, and in particular to codes over finite quasi-Frobenius rings (see also~\cite{Flanagan_cw_ind}). It was shown that the connections between LP decoding and MP decoding are preserved in the nonbinary settings as well. 

A promising approach for constructing LDPC codes using graphs goes back to~\cite{Tanner1}. 
The construction was modified in~\cite{Spielman}, where~\emph{expander graph} was used as 
an ingredient in a construction of linear-time decodable codes that correct a constant fraction of errors
under a variation of an MP decoder.  
This result was 
improved in the works~\cite{Zemor03},~\cite{Zemor04},~\cite{Roth-Skachek},~\cite{Skachek-Roth-2003},~\cite{Zemor01}. 
It was shown in~\cite{Zemor02} that expander codes achieve capacity of a binary symmetric channel
under a variation of MP decoder. 
Explicit constructions of regular expander graphs can be found, for instance, in~\cite{LPS},~\cite{Margulis}.  

In~\cite{Feldman-fraction}, the performance of expander codes in~\cite{Spielman} 
under the LP decoding was investigated. It was shown, that the LP decoder corrects a similar fraction of errors
as the MP decoder in~\cite{Spielman} does. This research direction was extended in~\cite{Feldman-capacity}, 
where it was shown that the expander codes achieve the capacity of a variety of 
binary memoryless channels. It was also shown in~\cite{Feldman-capacity}, that the LP decoder 
applied to the codes in~\cite{Zemor01} corrects a similar fraction of errors as the decoder therein,
which is approximately a quarter of the lower bound on their relative minimum distance. 

In this work, we generalize several results in~\cite{Feldman-capacity} toward \emph{nonbinary} settings. 
There are some additional differences between~\cite{Feldman-capacity} and our work. 
First, we use a slightly different definition of a (bipartite) expander graph and corresponding code. 
Second, the analysis in~\cite{Feldman-capacity} assumes that the all-zero codeword was transmitted,
while we do not make such an assumption. 
Finally, we present a more accurate analysis of the correctable fraction of errors, and, in particular, we
elaborate on the $o(1)$-term in the bound on a fraction of correctable errors. 

The manuscript is structured as follows. In Section~\ref{sec:defs}, we redefine (nonbinary) expander codes. 
In Section~\ref{sec:lp-decoder} we define a linear-programming decoder for these codes and discuss 
some of its basic properties. In Section~\ref{sec:dual}, we present the dual problem and 
discuss the criteria for the decoding success. In Section~\ref{sec:constant-fraction}, 
we present a feasible solution to the dual problem and show that the LP decoder corrects 
a constant fraction of errors. In Section~\ref{sec:error-orientation}, we present a concept 
of error pattern orientation. By using this concept, we show that the LP decoder corrects 
even higher fraction of errors. Finally, in Section~\ref{sec:discussion}, we summarize the results 
presented in this paper and compare them with some known works.

%--New------------------------------

\section{Code Construction}
\label{sec:defs}

Below, we revisit the construction in~\cite{Zemor02}. 

Let $\graph = (A \cup B,E)$ be 
a bipartite $\Delta$-regular undirected connected graph with
a vertex set $V = A \cup B$ such that $A \cap B = \emptyset$,
and an edge set $E$ such that
every edge has one endpoint in $A$ and one endpoint in $B$.
We denote $|A|=|B|=n$ and thus $|E| = \Delta n$.
We assume an ordering on $V$, thereby inducing
an ordering on $E = \{ e_i \}_{i=1}^{|E|}$.
Let $\ff$ be the field $\ffq$. 
For every vertex $v \in V$, we denote by $E(v)$ the set of
edges that are incident with $v$.
For a word $\bldz = (z_e)_{\!\scriptscriptstyle e \in E}$ 
(whose entries are indexed by $E$) in $\ff^{|E|}$,
we denote by $(\bldz)_{\!\scriptscriptstyle E(v)}$
the sub-block of $\bldz$ that is indexed by $E(v)$.

For each $v \in V$, let $\code(v)$ 
be a linear code of length $\Delta$ over $\ff$. 
The expander code $\Code$ is defined 
as the following linear code of length $|E|$ over $\ff$:
$$
\Code = \left\{ \bldc \in \ff^{|E|} \,:\,
\textrm{$(\bldc)_{\!\scriptscriptstyle E(v)} \in \code(v)$
for every $v \in V$} \right\} \; . 
$$

Suppose that $\code_A$ and $\code_B$ are
linear $[\Delta,r_A \Delta, \delta_A \Delta]$ and 
$[\Delta,r_B \Delta,\delta_B \Delta]$ codes over $\ff$, respectively.   
In the sequel, we consider the code $\Code$ with 
\[
\code(v) = \left\{ \begin{array}{cc}
\code_A & \mbox{ for every } v \in A \\
\code_B & \mbox{ for every } v \in B 
\end{array} \right. \; . 
\]
This code was first studied in~\cite{Zemor02}. 
In particular, it was shown therein that the rate of $\Code$
is at least $r_A + r_B - 1$. 

Denote by $A_\graph$ the adjacency matrix of $\graph$; namely,
$A_\graph$ is a $|V| \times |V|$ real symmetric matrix
whose rows and columns are indexed by the set $V$, and for every 
$u, v \in V$, the entry in $A_\graph$ that is indexed by $(u,v)$
is given by 
$$
(A_\graph)_{u,v} =
\left\{
\begin{array}{lcl} 
1 && \textrm{if $\{u, v\} \in E$} \\
0 && \textrm{otherwise} \\
\end{array}
\right. \; .
$$
It is known that $\Delta$ is the largest eigenvalue of $A_\graph$.
We denote by $\gamma_\graph$ the ratio 
between the second largest eigenvalue of $A_\graph$ and $\Delta$.
The constructions of $\Delta$-regular bipartite expander graphs in~\cite{LPS},~\cite{Margulis}
have $\gamma_\graph \le 2\sqrt{\Delta-1}/\Delta$. 

The relative minimum distance of $\Code$, $\delta_\Code$,  
was shown in~\cite{Roth-Skachek} to satisfy
\begin{equation}
\delta_\Code \ge \frac{\delta_A \delta_B - \gamma_\graph \sqrt{ \delta_A \delta_B}}
{1 - \gamma_\graph} \; .
\label{eq:min-distance}
\end{equation}

In the sequel, we use the notation $\distance(\bldx, \bldz)$ to denote the 
Hamming distance between the vectors $\bldx$ and $\bldz$.

\section{Linear-programming decoder}
\label{sec:lp-decoder}

In this section, we introduce an LP decoder for the code $\Code$.  
Suppose that the codeword $\bldc = (c_e)_{e \in E} \in \Code$ is transmitted through the adversarial channel
and the word $\bldy = (y_e)_{e \in E} \in \ff^{|E|}$ is received.

We define the mapping 
\[
\bldxi \; : \; \ff \longrightarrow \{ 0, 1 \}^{q} \subset \mathbb{R}^{q} \; , 
\]
by
\[
\bldxi (\beta) = \bldx = ( x^{(\alpha)} )_{\alpha \in \ff } \; , 
\]
such that, for each $\alpha \in \ff$,
\[
x^{(\alpha)}=\left\{ \begin{array}{cc}
1 & \textrm{ if } \alpha = \beta \\
0 & \textrm{ otherwise. }\end{array}\right. \;
\]
The mapping $\bldxi$ is one-to-one, 
and its image is the set of binary vectors of length $q$ with Hamming weight 1. 
Please note that this mapping is slightly different from its counterpart in~\cite{FSBG}, 
where the image of the mapping was the set of binary vectors of length $q-1$ of Hamming weight 0 or 1. 

We also define
\[
\bldXi \; : \; \ff^{|E|} \longrightarrow \{ 0, 1 \}^{q|E|} \subset \mathbb{R}^{q|E|} \; , 
\]
according to
\[
\bldXi(\bldc) = ( \bldxi(c_{e_{1}}) \; | \; 
\bldxi(c_{e_{2}}) \; | \; \cdots \; | \; \bldxi(c_{e_{|E|}})) \; .
\]
We note that $\bldXi$ is also one-to-one.

For vectors $\bldf \in \mathbb{R}^{q|E|}$, we adopt the notation
\[
\bldf = ( \bldf_{e_{1}} \; | \; \bldf_{e_{2}} \; | \; \cdots  \; | \; \bldf_{e_{|E|}} ) \; , 
\]
where
\[
\forall e \in E , \; \bldf_e = ( f_e^{(\alpha)} )_{\alpha \in \ff} \; .
\]
We can write the inverse of $\bldXi$ as
\[
\bldXi^{-1} (\bldf) = ( \bldxi^{-1}(\bldf_{e_{1}}), \bldxi^{-1}(\bldf_{e_{2}}), \cdots, 
\bldxi^{-1}(\bldf_{e_{|E|}}) ) \; .
\]

Below, we define the variables that will be used in the decoder. For all $e \in E$, $\alpha \in \ff$, 
we use the variables $f_e^{(\alpha)} \ge 0$. The objective function is 
$\sum_{e \in E} \sum_{\alpha \in \ff} \gamma_e^{(\alpha)} f_e^{(\alpha)}$, 
where $\gamma_e^{(\alpha)}$ is a function of the channel output. 

For each $\alpha \in \ff$ we set 
\[
\gamma_e^{(\alpha)} = \left\{ \begin{array}{ll}
- 1 & \mbox{if } \alpha = y_e \\
1 & \mbox{if } \alpha \neq y_e
\end{array} \right. \; . 
\]

Assume that $\bldf_e = \bldxi(\beta)$ for some $e\in E$, $\beta \in \ff$. Then,
it is straightforward to verify that 
\[
\sum_{\alpha \in \ff} \gamma_e^{(\alpha)} f_e^{(\alpha)} =
\left\{ \begin{array}{cl}
- 1 & \mbox{ if }  \beta = y_e \\
1 & \mbox{ if } \beta \neq y_e
\end{array} \right. \; .  
\] 

Suppose now that $\bldf = \Xi(\bldz)$ for some $\bldz \in \ff^{|E|}$. 
It follows that 
\begin{equation}
\sum_{e \in E} \sum_{\alpha \in \ff} \gamma_e^{(\alpha)} f_e^{(\alpha)} + |E| = 2 \distance(\bldy, \bldz)  \; .
\label{eq:hamming-dist}
\end{equation}
(Recall that the notation $\distance(\bldy, \bldz)$ is used for 
the Hamming distance between $\bldy$ and $\bldz$.)
Therefore, finding $\bldz \in \Code$ such that $\bldf = \Xi(\bldz)$ 
minimizes the left-hand side of~(\ref{eq:hamming-dist}) is equivalent to the nearest-neighbor decoding of $\bldy$.  
Instead, however, we will equivalently maximize
\begin{equation}
- \sum_{e \in E} \sum_{\alpha \in \ff} \gamma_e^{(\alpha)} f_e^{(\alpha)} \; . 
\label{eq:objective-func}
\end{equation}

In the sequel, we use the variables $w_{v,\bldbb}$ for all $v \in V$ and all $\bldb \in \code(v)$. 
These variables can be viewed as \emph{relative weights} of local codewords $\bldb$ associated with the
edges incident with the vertex $v$. The corresponding linear-programming problem is presented in 
Figure~\ref{fig:lp-primal}. 

\begin{figure*}[htb]
\small \noindent
\makebox[0in]{}\hrulefill\makebox[0in]{} 
\begin{eqnarray}
\mbox{\sf\footnotesize Maximize} \qquad \qquad \phantom{\forall v \in V} & & \sum_{e \in E, \alpha \in \ff} \left( - \gamma_e^{(\alpha)} \right) \cdot  f_e^{(\alpha)}  \label{LP-primal-1} \\
\mbox{\sf\footnotesize subject to} \qquad \qquad \forall v \in V & : & \sum_{\bldbb \in \code(v)} w_{v,\bldbb} = 1 \; ; \label{LP-primal-2}\\
\forall e = \{ v,u \} \in E, \; \forall \alpha \in \ff & : & f_e^{(\alpha)} = \sum_{\bldbb \in \code(v) \; : \; b_e = \alpha} w_{v,\bldbb} \; , \label{LP-primal-3} \\
& & f_e^{(\alpha)} = \sum_{\bldbb \in \code(u) \; : \; b_e = \alpha} w_{u,\bldbb} \; ; \label{LP-primal-4} \\
\forall e \in E, \; \alpha \in \ff & : & f_e^{(\alpha)} \ge 0 \; ; \\
\forall v \in V, \; \bldb \in \code(v) & : & w_{v,\bldbb} \ge 0 \; \label{LP-primal-6} .  
\end{eqnarray}
\makebox[0in]{}\hrulefill\makebox[0in]{}
\caption{Primal LP problem}
\label{fig:lp-primal}
\end{figure*}

Constraints~(\ref{LP-primal-2})-(\ref{LP-primal-6}) form a polytope which we denote by $\cQ$. 
In particular, it follows from constraints~(\ref{LP-primal-2})-(\ref{LP-primal-6}) that 
\begin{equation}
\forall e \in E \; : \; \sum_{\alpha \in \ff} f_e^{(\alpha)} = 1 \; . 
\label{eq:sum-1}
\end{equation}

Next, we define the decoding algorithm for the code $\Code$. The decoder optimizes the objective 
function~(\ref{LP-primal-1}) subject to constraints~(\ref{LP-primal-2})-(\ref{LP-primal-6}). 
If the result $\bldf$ is in $\{0,1\}^{q|E|}$, then the decoder outputs $\Xi^{-1}(\bldf)$
(as it is shown below, this output is then a codeword of $\Code$). Otherwise, the decoder 
declares a \emph{decoding failure}. 

We have the following proposition. 
\begin{proposition}
\label{prop:equivalence}
$\,$
\begin{enumerate}
\item
Let $(\bldf,\bldw) \in \cQ$ and $\bldf \in \{0, 1\}^{q|E|}$. Then  
\[
\bldXi^{-1} (\bldf) \in \Code \; . 
\] 
\item
If $\bldc \in \Code$ then there exists $\bldw$ such that $(\bldf, \bldw) \in \cQ$ and 
$\bldf = \Xi(\bldc) \in \{0,1\}^{q|E|}$.  
\end{enumerate}
\end{proposition}

\begin{proof} 
\begin{enumerate}
\item
Suppose $(\bldf,\bldw) \in \cQ$ and $\bldf \in \{0, 1\}^{q|E|}$. Let $\bldc = \Xi^{-1} (\bldf)$. 
By~(\ref{eq:sum-1}), $\bldc$ is well defined. 
Next, fix some $v \in V$ and let $\blda = (\bldc)_{\!\scriptscriptstyle E(v)}$ (for $\blda = (a_e)_{e \in E(v)})$. It follows that for any $e \in E(v)$, $\alpha \in \ff$, 
$f_e^{(\alpha)} = 1$ if and only if $a_e = \alpha$. Let $\bldd \in \code(v)$, $\bldd \neq \blda$. 
Since $\blda$ and $\bldd$ are different, there exists $\beta \in \ff$ and $e' \in E(v)$ such that 
$a_{e'} \neq \beta$ and $d_{e'} = \beta$. 
Then, it follows from~(\ref{eq:sum-1}) and either from~(\ref{LP-primal-3}) or from~(\ref{LP-primal-4}) that 
\[
0 = f_{e'}^{(\beta)} = \sum_{\bldbb \in \code(v) \; : \; b_{e'} = \beta} w_{v,\bldbb} \; , 
\]
and therefore $w_{v,\blddd} = 0$. 

It follows that $w_{v,\blddd}=0$ for all $\bldd \in \code(v)$, $\bldd \neq \blda$, and that $w_{v,\bldaa}=1$. 
Applying this argument 
for every $v \in V$ implies $\bldc \in \Code$. 
 
\item
Assume that $\bldc \in \Code$. Let $\bldf = \Xi(\bldc)$. For each $v \in V$, we set   
\[
w_{v,\bldbb} = \left\{ \begin{array}{cl}
1 & \mbox{ if } \bldb = (\bldc)_{\!\scriptscriptstyle E(v)} \\
0 & \mbox{ otherwise } 
\end{array} \right. \; . 
\]
The reader can easily verify that $\bldf \in \{0,1\}^{q|E|}$ and 
the corresponding $(\bldf, \bldw)$ is in~$\cQ$. 
\end{enumerate}
\end{proof} 

The following theorem is an equivalent of the \emph{nearest-neighbor certificate}. 
\begin{theorem}
Suppose that the LP solver applied to the LP problem in Figure~\ref{fig:lp-primal} 
outputs a codeword $\bldc \in \Code$. Then, $\bldc$ is the nearest-neighbor codeword. 
\end{theorem}
The proof follows from the previous proposition and~(\ref{eq:hamming-dist}).

\section{Dual Witness and Unique Solution}
\label{sec:dual} 

We aim to show that the decoder succeeds given that the number of adversarial errors 
is bounded from above by a certain constant. We use the \emph{dual witness} approach proposed 
in~\cite{Feldman-fraction}. This technique was extended in~\cite{Feldman-capacity} 
toward binary expander code. We further extend this technique toward nonbinary settings. 

Recall that the codeword $\bldc \in \Code$ was transmitted. If that is the case,
the decoder succeeds if it outputs the same $\bldc$. It follows from 
Proposition~\ref{prop:equivalence} that there is only one feasible combination of values of 
the variables $w_{v,\bldbb}$ that corresponds to the codeword $\bldc$, namely 
\[
\forall v \in V \; : \; w_{v,\bldbb} = \left\{ \begin{array}{cl}
1 & \mbox{ if } \bldb = (\bldc)_{\!\scriptscriptstyle E(v)} \\
0 & \mbox{ otherwise } 
\end{array} \right. \; . 
\]
The sufficient criteria 
for the decoder success is that this solution is the {\em unique} optimum of the 
LP decoding problem in Figure~\ref{fig:lp-primal}. 

To prove the optimality, we show the existence of 
a dual feasible solution, such that the value of the objective function
of the dual problem is equal to the value of the objective functions of the primal problem. 
The dual LP problem makes use of the following variables. For each $\alpha \in \ff$, $e \in E$, and $v \in V$,  
such that $v$ is an endpoint of $e$, there is a variable $\tau_{v,e}^{(\alpha)}$. In addition,
for each $v \in V$, there is a variable $\sigma_v$. 

\begin{figure*}[htb]
\small \noindent
\makebox[0in]{}\hrulefill\makebox[0in]{} 
\begin{eqnarray}
\mbox{\sf\footnotesize Minimize} \qquad \phantom{\; \forall e = \{ v,u \} \in E, \; \forall \alpha \in \ff} 
&& \sum_{v \in V} \sigma_v  \label{LP-dual-1} \\
\mbox{\sf\footnotesize subject to} \qquad \forall e = \{ v,u \} \in E, \; \forall \alpha \in \ff & : & 
\tau_{v,e}^{(\alpha)} + \tau_{u,e}^{(\alpha)} \le \gamma_e^{(\alpha)} \; ; \label{LP-dual-2} \\
 \forall v \in V, \; \forall \bldb \in \code(v) & : & \sum_{e \in E(v)} \tau_{v,e}^{(b_e)} 
+ \sigma_v \ge 0 \;  \label{LP-dual-3} .  
\end{eqnarray}
\makebox[0in]{}\hrulefill\makebox[0in]{}
\caption{Dual LP problem}
\label{fig:lp-dual}
\end{figure*}

The dual LP problem is presented in Figure~\ref{fig:lp-dual}. 
We set the objective value to be $|E| - 2 \dist(\bldy, \bldc)$, which is the value in~(\ref{eq:objective-func}) 
under the substitution $\bldz = \bldc$ (this fact easily follows from~(\ref{eq:hamming-dist})).  
This can be achieved by setting, for all $v \in V$, $\sigma_v = \half \Delta - \dist((\bldy)_{\!\scriptscriptstyle E(v)}, 
(\bldc)_{\!\scriptscriptstyle E(v)})$. 

In order to show the uniqueness of the solution, we slightly modify the dual LP problem. 
More specifically, we enforce strict inequalities in~(\ref{LP-dual-2}), 
such that the corresponding dual polytope (denoted by $\cP$)
becomes as in Figure~\ref{fig:lp-dual-slack}. Generally speaking, the polytope $\cP$ 
can be unbounded, and thus, sometimes we use the term ``open polytope''.  
\begin{figure*}[htb]
\small \noindent
\makebox[0in]{}\hrulefill\makebox[0in]{} 
\begin{eqnarray}
\hspace{-5ex} 
\forall e = \{ v,u \} \in E, \; \forall \alpha \in \ff \backslash \{ c_e \} & : & 
\tau_{v,e}^{(\alpha)} + \tau_{u,e}^{(\alpha)} < \gamma_e^{(\alpha)} \; ; \label{LP-polytope-1} \\
\forall e = \{ v,u \} \in E  & : & 
\tau_{v,e}^{(c_e)} + \tau_{u,e}^{(c_e)} \le \gamma_e^{(c_e)} \; ; \label{LP-polytope-1-1} \\
\forall v \in V, \; \forall \bldb \in \code(v) & : & \sum_{e \in E(v)} \tau_{v,e}^{(b_e)} 
\ge  - \half \Delta + \distance ( (\bldy)_{\!\scriptscriptstyle E(v)}, (\bldc)_{\!\scriptscriptstyle E(v)}) \; . 
\label{LP-polytope-2}    
\end{eqnarray}
\makebox[0in]{}\hrulefill\makebox[0in]{}
\caption{Dual (open) polytope $\cP$}
\label{fig:lp-dual-slack}
\end{figure*}

The uniqueness of the solution for the primal LP problem now follows from the following proposition. 

\begin{proposition}
If there is a feasible point in the polytope $\cP$, then there is a \emph{unique} 
optimum for the primal LP problem in Figure~\ref{fig:lp-primal}.
\label{prop:unique}
\end{proposition}

\begin{proof} 
First, it is straight-forward to see that any feasible point 
$\bldtau = \{ \tau_{v,e}^{(\alpha)} \}_{v \in V, e \in E, \alpha \in \ff}$ 
in $\cP$ is also a feasible point in the polytope in Figure~\ref{fig:lp-dual}
with $\sigma_v = \half \Delta - \distance((\bldy)_{\!\scriptscriptstyle E(v)},(\bldc)_{\!\scriptscriptstyle E(v)})$, for all $v \in V$.  
Then, it follows from~(\ref{eq:hamming-dist}) that $(\bldf, \bldw)$ is an optimal solution
for the primal problem in Figure~\ref{fig:lp-primal}, where 
\[
\forall e \in E \; : \; \bldf_e = \xi(c_e) \; . 
\]
Assume that $(\bldh, \blds)$ is another optimal solution
for the LP problem in Figure~\ref{fig:lp-primal}. 

Inequality~(\ref{LP-polytope-1}) implies that 
\[
\tau_{v,e}^{(\alpha)} + \tau_{u,e}^{(\alpha)} \le \gamma_e^{(\alpha)} - \varepsilon \; , 
\]
for some small $\varepsilon > 0$, for all $e = \{ v, u \} \in E$, $\alpha \in \ff \backslash \{ c_e \}$. 
We define a new cost function $\hat{\bldgamma} = \{ \hat{\gamma}_e^{(\alpha)} \}_{e \in E, \alpha \in \ff}$ 
for the problem in Figure~\ref{fig:lp-primal} as
follows:
\[
\hat{\gamma}_e^{(\alpha)} = \left\{ \begin{array}{cl}
\gamma_e^{(\alpha)} - \varepsilon & \mbox{ if } f_e^{(\alpha)} = 0 \\
\gamma_e^{(\alpha)} & \mbox{ otherwise } 
\end{array} \right. \; . 
\]
Observe, that 
\[
\sum_{e \in E, \alpha \in \ff} \left( - \hat{\gamma}_e^{(\alpha)} \right) \cdot  f_e^{(\alpha)}  = 
\sum_{e \in E, \alpha \in \ff} \left( - \gamma_e^{(\alpha)} \right) \cdot  f_e^{(\alpha)} \; . 
\]
It follows that $(\bldf, \bldw)$ is an optimal solution
for the LP problem in Figure~\ref{fig:lp-primal} under the cost function $\hat{\bldgamma}$. 

Note that $(\bldf, \bldw)$ corresponds to a codeword $\bldc$, and so its entries are either $0$ or $1$. 
Moreover, $(\bldf, \bldw) \neq (\bldh, \blds)$, and so in particular $\bldf \neq \bldh$. 
Therefore, there must exist at least one $e \in E$ such that $\bldf_e \neq \bldh_e$. 
For such $e$, due to~(\ref{eq:sum-1}) (with respect to $\bldh_e$), there exists at 
least one $\beta \in \ff$ such that 
$f_e^{(\beta)} = 0$ and $h_e^{(\beta)} > 0$. Therefore, 
\begin{eqnarray*}
\sum_{e \in E, \alpha \in \ff} \left( - \hat{\gamma}_e^{(\alpha)} \right) \cdot h_e^{(\alpha)}
> \sum_{e \in E, \alpha \in \ff} \left( - \gamma_e^{(\alpha)} \right) \cdot h_e^{(\alpha)} \phantom{\; ,} \\
= \sum_{e \in E, \alpha \in \ff} \left( - \gamma_e^{(\alpha)} \right) \cdot f_e^{(\alpha)} \phantom{\; ,} \\
= \sum_{e \in E, \alpha \in \ff} \left( - \hat{\gamma}_e^{(\alpha)} \right) \cdot f_e^{(\alpha)} \; ,
\end{eqnarray*}
and this makes a contradiction to the fact that $(\bldf, \bldw)$ is an optimal solution 
to the primal problem under the cost function $\hat{\bldgamma}$. 
The contradiction follows from the (false) assumption that there is 
more than one optimal solution
for the original primal problem. 
\end{proof}

The following corollary follows immediately from Pro\-po\-si\-tion~\ref{prop:unique}. 
\begin{corollary}
If there is a feasible point in the polytope $\cP$, then the decoder in Figure~\ref{fig:lp-primal}
succeeds.
\end{corollary}

\section{Correcting a Constant Fraction of Errors} 
\label{sec:constant-fraction} 

Recall that the word $\bldc = (c_e)_{\!\scriptscriptstyle e \in E}\in \Code$ was transmitted
and $\bldy = (y_e)_{\!\scriptscriptstyle e \in E} \in \ff^{|E|}$
was received. Suppose that $\graph = (A \cup B, E)$ is 
a $\Delta$-regular bipartite graph defined as in Section~\ref{sec:defs}.

In this section, we will define a notion of error core. Building on that, 
we will show that if there is no error core in the graph $\graph$, 
then the dual solution can be always found for the appropriate nonbinary 
LP decoding problem. 

{\em Definition:}
The graph $\graph$ has an $(\zeta_A, \zeta_B)$-error core 
(where $\zeta_A, \zeta_B \in [0,1]$) associated with the word $\bldy$
if there exists a subset of edges in error $E' \subseteq \{ e \in E \; : \; y_e \neq c_e \}$ 
and two subsets of vertices $A' \subseteq A$ and $B' \subseteq B$ such that $A' \cup B'$ 
is the set of all the endpoints of the edges in $E'$, and:
\begin{itemize}
\item
for any $v \in A'$: $| \{ E(v) \cap E' \} | \ge \zeta_A \Delta$;
\item
for any $v \in B'$: $| \{ E(v) \cap E' \} |  \ge \zeta_B \Delta$. 
\end{itemize}

Below, we inductively define the sets of vertices $V_i$ 
(for $i = 0, 1, \cdots, t$, where $t$ will be defined later)
and the sets of edges $E_i$ (for $i = 1, 2, \cdots, t$) as follows. 
\begin{itemize}
\item {\em Basis.} The edge set $E_1$ will be the set of all edges corresponding 
to the erroneous symbols in~$\bldy$, and the vertex sets $V_0$ and $V_1$ will be the endpoints
of edges in~$E_1$:
\begin{eqnarray*}
E_1 & = & \{ e \in E \; : \; y_e \neq c_e \} \; ; \\ 
V_0 & = & \{ v \in A \; : \; E(v) \cap E_1 \neq \emptyset \} \; ; \\ 
V_1 & = & \{ v \in B \; : \; E(v) \cap E_1 \neq \emptyset \} \; . 
\end{eqnarray*}
\item {\em Step.} For $i \ge 2$:
\begin{eqnarray*}
V_i = \Big\{ v \in V_{i-2} \, : \,
\Big| \{ e \in E(v) \cap E_{i-1} \} \Big|
\ge \frac{\delta \Delta}{4} \Big\} , 
\end{eqnarray*}
where $\delta = \delta_A$ if $i$ is even, and $\delta = \delta_B$ if $i$ is odd, and
\[
E_i = \Big\{ e= \{v, u\} \in E_{i-1} \, : \, v \in V_{i-1}, \; u \in V_i \Big \} . 
\]
\end{itemize}

\begin{lemma}
If $E_i = \emptyset$ for some finite $i$, then the decoder in Figure~\ref{fig:lp-primal} succeeds. 
\label{lemma:main-1}
\end{lemma}

\begin{proof} 
We show that the decoder succeeds by constructing a feasible point in the polytope $\cP$. 
We use $\epsilon > 0$ to denote the quantity, which can be made as small as desired. 
The precise value of $\epsilon$ will be discussed later. 
We set the variables $\tau_{u,e}^{(\alpha)}$ as follows. 
\begin{itemize}
\item 
Let $e = \{v,u\} \notin E_1$. Then, by definition of $E_1$, $c_e = y_e$. Assume that $c_e = \beta$.  
We set, $\tau_{v,e}^{(\beta)} = \tau_{u,e}^{(\beta)} = -1/2$, and 
so $\tau_{v,e}^{(\beta)} + \tau_{u,e}^{(\beta)} \le \gamma_e^{(\beta)} = -1$.  We also set 
$\tau_{v,e}^{(\alpha)} = \tau_{u,e}^{(\alpha)} = 1/2 - \epsilon$
for all $\alpha \in \ff \backslash \{ \beta \}$. 
In that case, $\tau_{v,e}^{(\alpha)} + \tau_{u,e}^{(\alpha)} < \gamma_e^{(\alpha)} = 1$.
Therefore,~(\ref{LP-polytope-1}) and~(\ref{LP-polytope-1-1}) are satisfied. 
\item
Let $e = \{v,u\} \in E_1$. Denote $c_e = \beta$. By definition of $E_1$, $y_e \neq c_e$. 
Let $i^*$ be the value such that $e \in E_{i^*} \backslash E_{i^*+1}$. 
In addition, without loss of generality assume that 
$v \in V_{i^*-1}$ and $u \in V_{i^*}$ (and so 
$v \notin V_{i^*+1}$ and 
$| E(v) \cap E_{i^*}| < \quarter \delta \Delta$). 

Then, we set $\tau_{v,e}^{(\beta)} = \tau_{u,e}^{(\beta)} = \half$. In that case, 
$\tau_{v,e}^{(\beta)} + \tau_{u,e}^{(\beta)} \le \gamma_e^{(\beta)} = 1$,
and so~(\ref{LP-polytope-1-1}) is satisfied.
We also set, for all $\alpha \in \ff \backslash \{ \beta \}$, 
$\tau_{v,e}^{(\alpha)} = - \fhalf - \epsilon$ and $\tau_{u,e}^{(\alpha)} = \thalf$, which yields
$\tau_{v,e}^{(\alpha)} + \tau_{u,e}^{(\alpha)} < \gamma_e^{(\alpha)} \in \{ -1, 1\}$.
Thus, all inequalities~(\ref{LP-polytope-1}) are also satisfied. 
\end{itemize}

Table~\ref{table:tau} summarizes the assignments of the values to variables $\tau_{v,e}^{(\alpha)}$
for all $e \in E$, $v \in e$ and $\alpha \in \ff$. 

\begin{table}[htb]
\begin{center}
{
%\scriptsize
$
\begin{array}{|c||c|c|}
\hline
&&\\
& \alpha = c_e & \alpha \neq c_e\\
&&\\
\hline
\hline 
&&\\
y_e \mbox{ is correct } & \tau_{v,e}^{(\alpha)} = - \half  & 
 \tau_{v,e}^{(\alpha)} = \half - \epsilon \\
&&\\
\hline 
&&\\
y_e \mbox{ is in error } & \tau_{v,e}^{(\alpha)} = \half  &  
\tau_{v,e}^{(\alpha)} = - \fhalf - \epsilon \mbox{ or } \tau_{v,e}^{(\alpha)} = \thalf  
\\
&& \mbox{ depends on the structure } \\
&& \mbox{ of the error } \\
&&\\
\hline 
\end{array}
$
}
\caption{Assignments of the values to the variables $\tau_{v,e}^{(\alpha)}$.}
\label{table:tau}
\end{center}
\end{table}

Since $E_i = \emptyset$ for some finite $i$ (we set $t = i+1$, where $i$ is this value), 
the values of all the variables $\tau_{v,e}^{(\alpha)}$ are defined. We already showed that all inequalities~(\ref{LP-polytope-1}) and~(\ref{LP-polytope-1-1})
are satisfied. Next, we show that inequalities~(\ref{LP-polytope-2}) are satisfied. It will be enough to show that 
for all $v \in V$, $\bldb \in \code(v)$,
\begin{eqnarray} 
\sum_{e \in E(v)} \tau_{v,e}^{(b_e)} 
\ge - \half \Delta + \distance((\bldy)_{\!\scriptscriptstyle E(v)}, (\bldc)_{\!\scriptscriptstyle E(v)})  \; . 
\label{eq:positive}
\end{eqnarray} 

For a vertex $v \in V$ and a codeword $\bldb \in \code(v)$, we define five sets of indices (edges) as follows:
\begin{eqnarray*}
\cE_1 & = & \{ e \in E(v) \; : \; y_e \mbox{ is correct and } b_e = c_e\} \; , \\ 
\cE_2 & = & \{ e \in E(v) \; : \; y_e \mbox{ is correct and } b_e \neq c_e \} \; , \\ 
\cE_3 & = & \{ e \in E(v) \; : \; y_e \mbox{ is in error and } b_e = c_e \} \; , \\
\cE'_4 & = & \{ e \in E(v) \; : \; y_e \mbox{ is in error, } \\
&& \hspace{10ex} b_e \neq c_e \mbox{ and } \tau_{v,e}^{(b_e)} = - \fhalf - \epsilon \} \; , \\ 
\cE''_4 & = & \{ e \in E(v) \; : \; y_e \mbox{ is in error, } \\
&& \hspace{15ex} b_e \neq c_e \mbox{ and } \tau_{v,e}^{(b_e)} = \thalf  \} \; . 
\end{eqnarray*}
(These sets depend on $v$ and $\bldb$, in addition to their dependence on $\bldc$ and $\bldy$. 
However, we write $\cE_j$ rather than $\cE_j(v, \bldb)$ for the sake of simplicity.) 

Then, 
\begin{eqnarray*}
&& \hspace{-5ex} \sum_{e \in E(v)} \tau_{v,e}^{(b_e)} \; = \; \sum_{e \in \cE_1} \tau_{v,e}^{(b_e)} + \sum_{e \in \cE_2}\tau_{v,e}^{(b_e)} + \sum_{e \in \cE_3}\tau_{v,e}^{(b_e)} \\
&& \hspace{15ex} + \sum_{e \in \cE'_4}\tau_{v,e}^{(b_e)} + \sum_{e \in \cE''_4}\tau_{v,e}^{(b_e)} \\
& = & \sum_{e \in \cE_1} (- \half) + \sum_{e \in \cE_2} (\half - \epsilon) + \sum_{e \in \cE_3} \half \\
&& \hspace{15ex} + \sum_{e \in \cE'_4} (- \fhalf - \epsilon) + \sum_{e \in \cE''_4} \thalf \\
& \ge & \Big( - \half \Delta + \distance((\bldy)_{\!\scriptscriptstyle E(v)}, (\bldc)_{\!\scriptscriptstyle E(v)}) \Big) + \sum_{e \in \cE_2} (1 - \epsilon)\\
&& \hspace{15ex}  + \sum_{e \in \cE'_4} (-3 - \epsilon) + \sum_{e \in \cE''_4} 1  \; . 
\end{eqnarray*} 
In order to prove~(\ref{eq:positive}), it will be enough to show that
\begin{equation}
|\cE_2| + |\cE''_4| \ge 3 |\cE'_4| + \epsilon (|\cE_2| + |\cE'_4|) \; . 
\label{eq:to-show} 
\end{equation}

We observe several cases. 
\begin{itemize} 
\item
Consider a vertex $v \in (A \backslash V_0) \cup (B \backslash V_1)$. 

Then, 
\begin{multline}
\Big| \{ e \in E(v) \; : \; y_e \neq c_e \} \Big| \\
= | \cE_3 | + | \cE'_4 | + | \cE''_4 | = 0  \; ,
\label{eq:vertex-v-1}
\end{multline}
and so~(\ref{eq:to-show}) is satisfied for any $\epsilon \le 1$. 

\item
Consider a vertex $v \in V_0 \cup V_1$. 
Let $\delta = \delta_A$ if $v \in A$, 
and $\delta = \delta_B$ if $v \in B$. 
Since $E_i = \emptyset$ for some $i \in \nn$, we have that $v \in V_{i^*-1} \backslash V_{i^*+1}$ 
for some $i^* \in \nn$. 
Therefore, 
\[
|E(v) \cap E_{i^*}| < \quarter \delta \Delta \; . 
\]
We can write, with respect to this $v$ and any $\bldb$, that 
\[
|\cE'_4 | \le \quarter (\delta - \epsilon') \Delta \; , 
\]
or,
\begin{equation} 
\delta \Delta \ge 4 |\cE'_4 | + \epsilon' \Delta\;,
\label{eq:cE-4}
\end{equation}
for some small $\epsilon'>0$. 
\begin{itemize}
\item
If $\bldb = (\bldc)_{\!\scriptscriptstyle E(v)}$, then obviously $|\cE_2|=|\cE'_4|=|\cE''_4|=0$, and so~(\ref{eq:to-show})  
holds. 
\item
If $\bldb \neq (\bldc)_{\!\scriptscriptstyle E(v)}$, then recall that the relative minimum distance of $\code(v)$ is
at least $\delta$. Therefore, $|\cE_2| + |\cE'_4| + |\cE''_4| \ge \delta \Delta$, 
and by using~(\ref{eq:cE-4}):
\[
|\cE_2| + |\cE''_4| \ge \delta \Delta - |\cE'_4| \ge  3 |\cE'_4| + \epsilon' \Delta \; .
\]
We see that~(\ref{eq:to-show}) holds for all $\epsilon \le \epsilon'$. 
\end{itemize}
\end{itemize}
We have shown that that in all cases, for sufficiently 
small $\epsilon$,~(\ref{LP-polytope-2}) holds, and therefore 
there exists a feasible point in $\cP$.  
\end{proof}

\begin{lemma}
If there is no $(\quarter \delta_A, \quarter \delta_B)$-error core, then $E_i = \emptyset$ for some $i \in \nn$. 
\label{lemma:main-2}
\end{lemma} 

\begin{proof} 
Suppose that there is no $i \in \nn$ such that $E_i = \emptyset$. 
Since for all $i \in \nn$, $E_{i+1} \subseteq E_i$, we have that there exists some even 
$i^* \in \nn$, such that for any $i \ge i^*$, 
$E_{i+1} = E_{i} \neq \emptyset$. This, in turn, means that $V_{i^*+2} = V_{i^*}$ 
and $V_{i^*+3}=V_{i^*+1}$. However, this implies 
(without loss of generality) that every $v \in V_{i^*+1}$ and $u \in V_{i^*+2}$
has at least $\quarter \delta_A \Delta$ and $\quarter \delta_B \Delta$ 
incident edges in $E_{i^*+1}$, respectively. It follows that the set of edges $E_{i^*+1}$ 
together with the sets $V_{i^*}$ and $V_{i^*+1}$
forms a $(\quarter \delta_A, \quarter \delta_B)$-error core. 
\end{proof}

\begin{corollary}
If the LP decoder in Figure~\ref{fig:lp-primal} fails, then there 
exists an $(\quarter \delta_A, \quarter \delta_B)$-error core 
associated with the word $\bldy$ in the graph $\graph$. 
\label{thrm:core}
\end{corollary}

The proof follows immediately from Lemmas~\ref{lemma:main-1} and~\ref{lemma:main-2}. 

Next, we show that the LP decoder in Figure~\ref{fig:lp-primal} corrects all the errors in $\bldy$ if the amount of errors 
in it is at most a fraction of the code length. Consider a subgraph $\cH = (U_A \cup  U_B, \Edge)$ of $\graph$ with 
$U_A \subseteq A$, $U_B \subseteq B$ and $\Edge \subseteq E$.
For a vertex $v \in U_A \cup U_B$ denote by $\deg_{\cH}(v)$ its degree in the graph $\cH$. 
We use the following known result. 

\begin{proposition}
Let $U_A$ and $U_B$
be subsets of sizes $|U_A| = \alp |A|$ and $|U_B| = \bet |B|$, respectively,
such that $\alp+\bet > 0$. 
Let $\Edge$ be the edge set induced by the vertex set $U_A \cup U_B$, and denote $\cH = (U_A \cup  U_B, \Edge)$. 
Then, 
\begin{eqnarray}
2 |\Edge| & = & \sum_{v \in U_A \cup U_B} \deg_{\cH}(v) \nonumber \\
& \le & 2 \left( \alp \bet + \gamma_\graph \sqrt{\alp (1 - \alp) \bet (1 - \bet)} \right) \Delta n \nonumber \\
& \le & 2 ((1 - \gamma_\graph) \alp \bet 
+ \gamma_\graph \sqrt{\alp \bet}) \Delta n \; .
\label{prop:av-degree}
\end{eqnarray}
\label{prop:degree}
\end{proposition}
This statement is equivalent to Proposition~3.3 in~\cite{Roth-Skachek}. The first inequality is 
obtained when the tighter inequality in Lemma 3.2 in~\cite{Roth-Skachek} is used in the proof of Proposition 3.3. 
If the graph is a Ramanujan expander as in~\cite{LPS},~\cite{Margulis}, then for fixed $\alp$ and $\bet$, by increasing $\Delta$ (and so by reducing $\gamma_\graph$), 
it is possible to make $|\Edge|/(\Delta n)$ as close to  
$(\alp \bet)$ as desired. 

By using Proposition~\ref{prop:degree}, we obtain the following theorem. 
\begin{theorem}
Assume that the size of error in $\bldy$ is less than
$$
\frac{\zeta_A \zeta_B - \gamma_\graph \sqrt{\zeta_A \zeta_B}}{1 - \gamma_\graph} \cdot \Delta n \; , 
$$ 
for some $\zeta_A, \zeta_B \in (0, 1]$, such that $\gamma_\graph \le \sqrt{\zeta_A \zeta_B}$. 
Then, the graph $\graph$ contains no $(\zeta_A, \zeta_B)$-error core associated with this $\bldy$. 
\label{thrm:core-size}
\end{theorem}

The proof of this theorem is along the same lines as the proof of Theorem~3.1 in~\cite{Roth-Skachek}. 
For the sake of completeness of the presentation, we place the sketch of the proof in Appendix. 

The main result of this section follows from Corollary~\ref{thrm:core} and Theorem~\ref{thrm:core-size}, 
and it appears in the following corollary. 

\begin{corollary}
If the size of error in $\bldy$ is less than
$$
\frac{\delta_A \delta_B / 16 - \gamma_\graph \sqrt{\delta_A \delta_B / 16}}{1 - \gamma_\graph} \cdot \Delta n \; , 
$$ 
and $\gamma_\graph \le \quarter \sqrt{\delta_A \delta_B}$, then the LP decoder in Figure~\ref{fig:lp-primal}
will correct all errors in $\bldy$.  
\end{corollary}

Observe, that the proposed LP decoder corrects any error pattern of size approximately $\delta_A \delta_B \Delta n / 16$, 
when the value of $\Delta$ is large enough.

\section{Using Error Pattern Orientation}
\label{sec:error-orientation}

In this section, we present more powerful decoder analysis than its counterpart 
in Section~\ref{sec:constant-fraction}. More specifically, by using \emph{error pattern orientation}, 
we are able to improve the fraction of correctable errors in Section~\ref{sec:constant-fraction} 
by approximately a factor of $4$. The idea of using error pattern orientation was 
proposed in~\cite{Feldman-capacity}. 

Let $\graph = (A \cup B, E)$ be a $\Delta$-regular bipartite graph as before, and 
let $\cH = (U_A \cup  U_B, \Edge)$ be a subgraph with $U_A \subseteq A$, $U_B \subseteq B$ and $\Edge \subseteq E$.
We start with the following definition. 

%Since $\alp, \bet \in [0, 1]$, we have $0 \le \sqrt{\alp \bet} - \alp \bet \le \quarter$. 
%Therefore, the statement of Proposition~\ref{prop:av-degree} can be relaxed as  
%\begin{equation}
%|\Edge| = \sum_{v \in U_A} \deg_{\cH}(v) = \sum_{v \in U_B} \deg_{\cH}(v) \le (\alp \bet + \quarter \gamma_\graph ) \Delta n %\; . 
%\label{eq:num-edges}
%\end{equation}

%By switching to the complementary sets, Proposition~\ref{prop:av-degree} can also be reformulated as 
%\begin{eqnarray}
%|\Edge| & \ge & 
%\left( \alp \bet - \gamma_\graph \sqrt{\alp (1 - \alp) \bet (1 - \bet)} \right) \Delta n \nonumber \\
%& \ge & ((1 + \gamma_\graph) \alp \bet - \gamma_\graph \sqrt{\alp \bet}) \Delta n \; . 
%\label{eq:num-edges-2}
%\end{eqnarray}

%Let $\cH = (U_A \cup U_B , \Edge)$ be a subgraph of a bipartite 
%graph $\graph = (A \cup B, E)$ (therefore $U_A \subseteq A$, $U_B \subseteq B$ and $\Edge \subseteq E$). 

{\em Definition:} 
The assignment of the directions to the edges of the subgraph 
$\cH = (U_A \cup  U_B, \Edge)$ is called an \emph{$(\rho_A, \rho_B)$-orientation} (for some $\rho_A, \rho_B \in (0, 1]$) 
if each vertex $v \in U_A$ and each vertex $v \in U_B$ has at most $\rho_A \Delta$ and 
$\rho_B \Delta$ incoming edges in $\Edge$, respectively. 
We will say that for the given 
assignment of the edge directions, $M$ edges are \emph{violating the 
$(\rho_A, \rho_B)$-orientation property at the vertex $v \in U_A$ ($v \in U_B$)} 
if $v$ has $\rho_A \Delta + M$ ($\rho_B \Delta + M$, respectively) incoming edges in $\Edge$.
We will also say that for the given 
assignment of the edge directions, $M$ edges are \emph{violating the 
$(\rho_A, \rho_B)$-orientation property in $\cH$} if $M$ is the smallest integer such that by removing $M$ edges from 
$\Edge$, the resulting $\cH$ will have a $(\rho_A, \rho_B)$-orientation.

\begin{lemma}
Let $\cH = (U_A \cup U_B , \Edge)$ be a subgraph of $\graph = (A \cup B, E)$
with $U_A \subseteq A$, $U_B \subseteq B$ and $\Edge \subseteq E$.
Assume that 
$$
|\Edge| \le \frac{\betaf \alphaf - \gamma_\graph \sqrt{\betaf \alphaf}}{1 - \gamma_\graph} \cdot \Delta n \; , 
$$ 
for some $\betaf, \alphaf \in (0, 1]$, such that $\gamma_\graph \le \sqrt{\betaf \alphaf}$, and $\half \betaf \Delta$, $\half \alphaf \Delta$ are 
both integers. 
Then, $\Edge$ contains an $(\betaf/2, \alphaf/2)$-orientation.  
\label{lemma:orientation}
\end{lemma}

\begin{proof} 
Assign directions to the edges in $\Edge$
such that the number of violations of an $(\betaf/2, \alphaf/2)$-orientation in $\cH$ is minimal. 
We will show that if for some $v \in U_A$ ($v \in U_B$) 
there are more than $\betaf \Delta /2$ ($\alphaf \Delta /2$, respectively) 
incoming edges, then it is possible to change 
the directions of the edges in the graph such that the number of edges violating the orientation property
will decrease. This will make a contradiction to the minimality of the number of 
orientation violations in the current assignment
of the edge directions. 

Denote by $\deg_{\mbox{\scriptsize in}} (v)$ the number of incoming edges (in $\cH$) of the vertex $v$. 
Recall that $\betaf \Delta$ and 
$\alphaf \Delta$ are even integers. We will use the following definitions.

{\em Definition:} A vertex $v \in U_A \cup U_B$ is called a \emph{heavy} vertex if it satisfies  
one of the following:
\begin{enumerate}
\item
$v \in U_A$ and $\deg_{\mbox{\scriptsize in}}(v) > \half \betaf \Delta$; 
\item
$v \in U_B$ and $\deg_{\mbox{\scriptsize in}}(v) > \half \alphaf \Delta$.
\end{enumerate}

{\em Definition:} A vertex $v \in U_A \cup U_B$ is called a \emph{full} vertex if it satisfies  
one of the following:
\begin{enumerate}
\item
$v \in U_A$ and $\deg_{\mbox{\scriptsize in}}(v) = \half \betaf \Delta$;
\item
$v \in U_B$ and $\deg_{\mbox{\scriptsize in}}(v) = \half \alphaf \Delta$.
\end{enumerate}

{\em Definition:} A vertex $v \in U_A \cup U_B$ is called a \emph{light} vertex if it satisfies  
one of the following:
\begin{enumerate}
\item
$v \in U_A$ and $\deg_{\mbox{\scriptsize in}}(v) < \half \betaf \Delta$;
\item
$v \in U_B$ and $\deg_{\mbox{\scriptsize in}}(v) < \half \alphaf \Delta$.
\end{enumerate}

Observe that the orientation property 
is not violated at the full and at the light vertices. 
Assume, by contrary, that there exists a heavy vertex in $U_A \cup U_B$. We show that 
it is possible to change the directions of the edges in $\Edge$ such that 
the total number of edges violating the orientation property in $\cH$ will decrease. 

Define a set of vertices $U$ to be the maximal set as follows: 
\begin{itemize} 
\item 
If $v \in U_A \cup U_B$ is heavy then $v \in U$. 
\item
If $u \in U_A \cup U_B$ is full and there is a direct edge from $u$ to $v$ for some $v \in U$, then $u \in U$. 
\end{itemize} 
The set $U$ is well defined. 

If there is an edge $(w, u)$ for some $w \notin U$ and $u \in U$, then $w$ is \emph{light} and 
there exists a path from $w$ to some \emph{heavy} vertex $v \in U$ (vertex $u$ can be full). 
Then, it is possible to flip the directions of all edges in the path, and thus to decrease the 
number of violations of the orientation property by $1$ (at the vertex $v$). 

Below, we assume that there is no edge $(w, u)$ for any $w \notin U$ and $u \in U$. 
Denote $U'_A = U \cap U_A$ and $U'_B = U \cap U_B$. Let $\Edge'$ be a set of edges in $\Edge$ 
having one endpoint in $U'_A$ and one endpoint in $U'_B$. Let $\alp = |U'_A|/n$ and $\bet = |U'_B|/n$. 
We have 
\begin{multline}
\half (\alp \betaf + \bet \alphaf) \Delta n  < |\Edge' | \le | \Edge | \\
\le \frac{\betaf \alphaf - \gamma_\graph \sqrt{\betaf \alphaf}}{1 - \gamma_\graph} \cdot \Delta n \; , 
\label{eq:edges-lwr-bnd}
\end{multline} 
where the first inequality is correct since there are only heavy and full vertices 
in $U'_A \cup U'_B$, and at least one of these vertices is heavy. 
The last inequality is given by the conditions of the lemma. 

Assume that 
the ratio between the number of directed edges in $\Edge'$ from $U'_A$ to $U'_B$ and 
the number of directed edges in $\Edge'$ from $U'_B$ to $U'_A$ is $\kappa > 0$. 
Then, 
\begin{multline}
\half \alp \betaf (1 + \kappa) \cdot \Delta n \le |\Edge' | \\
\le \left( (1 - \gamma_\graph) \alp \bet  + \gamma_\graph \sqrt{\alp \bet} \right) \Delta n \; ,
\label{eq:edges-1}
\end{multline}
and 
\begin{multline}
\half \bet \alphaf (1 + 1/\kappa) \cdot \Delta n \le |\Edge' | \\
\le \left( (1 - \gamma_\graph) \alp \bet + \gamma_\graph \sqrt{\alp \bet} \right) \Delta n \; ,
\label{eq:edges-2}
\end{multline}
where the left-hand side inequalities follow from the fact that every vertex in $U'_A$ and 
every vertex in $U'_B$ is either full or heavy, and the right-hand side inequalities follow from~(\ref{prop:av-degree}). 

Inequalities~(\ref{eq:edges-1}) and~(\ref{eq:edges-2}) yield
\begin{equation}
\bet \ge \frac{\betaf(1 + \kappa)}{2 (1 - \gamma_\graph)} - \frac{\gamma_\graph}{1 - \gamma_\graph} \sqrt{ \frac{\bet}{\alp}}
\; , 
\label{eq:bet-lower-bnd}
\end{equation}
and 
\begin{equation}
\alp \ge \frac{\alphaf(1 + 1/\kappa)}{2 (1 - \gamma_\graph)} - \frac{\gamma_\graph}{1 - \gamma_\graph} \sqrt{ \frac{\alp}{\bet}}
\; , 
\label{eq:alp-lower-bnd}
\end{equation}
respectively. 

Consider two cases.
\begin{itemize}
\item[Case 1:]
$\alp \betaf (1 + \kappa) \ge \bet \alphaf (1 + 1/\kappa)$. 
Then, from~(\ref{eq:bet-lower-bnd}) we have
\[
\bet \ge \frac{\betaf(1 + \kappa)}{2(1 - \gamma_\graph)} - \frac{\gamma_\graph}{1 - \gamma_\graph} \sqrt{\frac{\betaf (1 + \kappa)}{\alphaf (1 + 1/\kappa)}} \; ,
\]
and, so,  
\[
\bet \alphaf \ge \frac{\betaf \alphaf (1 + \kappa)}{2(1 - \gamma_\graph)} - \frac{\gamma_\graph}{1 - \gamma_\graph} 
\sqrt{\betaf \alphaf \kappa} \; . 
\]
Finally, 
\begin{eqnarray*}
\alp \betaf & \ge & \bet \alphaf \frac{1 + 1/\kappa}{1 + \kappa} \\
& \ge & \frac{\betaf \alphaf (1 + 1/\kappa)}{2(1 - \gamma_\graph)} - \frac{\gamma_\graph}{1 - \gamma_\graph} 
\sqrt{\frac{\betaf \alphaf}{\kappa}} \; . 
\end{eqnarray*}
\item[Case 2:]
$\alp \betaf (1 + \kappa) < \bet \alphaf (1 + 1/\kappa)$. 
Then, from~(\ref{eq:alp-lower-bnd}) we have
\[
\alp > \frac{\alphaf(1 + 1/\kappa)}{2(1 - \gamma_\graph)} - \frac{\gamma_\graph}{1 - \gamma_\graph} 
\sqrt{\frac{\alphaf (1 + 1/\kappa)}{\betaf (1 + \kappa)}} \; ,
\]
and, so,  
\[
\alp \betaf > \frac{\betaf \alphaf (1 + 1/\kappa)}{2(1 - \gamma_\graph)} - \frac{\gamma_\graph}{1 - \gamma_\graph} 
\sqrt{\frac{\betaf \alphaf}{\kappa}} \; . 
\]
We also obtain: 
\begin{eqnarray*}
\bet \alphaf & > & \alp \betaf \frac{1 + \kappa}{1 + 1/\kappa} \\
& > & \frac{\betaf \alphaf (1 + \kappa)}{2(1 - \gamma_\graph)} - \frac{\gamma_\graph}{1 - \gamma_\graph} 
\sqrt{\betaf \alphaf\kappa} \; . 
\end{eqnarray*}
\end{itemize} 

From~(\ref{eq:edges-lwr-bnd}), in both cases we have: 
\begin{eqnarray}
|\Edge | & > & \frac{1}{2} (\alp \betaf + \bet \alphaf) \Delta n \nonumber \\
& \ge & \frac{1}{2} \Bigg( \frac{\betaf \alphaf (2 + \kappa + 1/\kappa)}{2(1 - \gamma_\graph)} \nonumber \\
&& \hspace{-1ex} - \; \frac{\gamma_\graph 
\sqrt{\betaf \alphaf}}{1 - \gamma_\graph} \left( \sqrt{\kappa} + \sqrt{\frac{1}{\kappa}} \right)\Bigg) \Delta n \, . 
\label{eq:eta}
\end{eqnarray}

Denote 
\[
\eta = \sqrt{\kappa} + \sqrt{1/\kappa}\;, \quad \eta \in [2, + \infty) \; . 
\]
Observe that the right-hand side of~(\ref{eq:eta}) is a quadratic function of $\eta$. 
Since $\gamma_\graph \le \sqrt{\betaf \alphaf}$, we have that 
this function is nonnegative and monotonic increasing for $\eta \ge 2 \gamma_\graph / \sqrt{\betaf \alphaf}$. 
Its minimum is obtained for the smallest value of $\eta$, which 
is achieved at $\kappa = 1$. Therefore,~(\ref{eq:eta}) becomes     
\[
|\Edge | > \frac{ \betaf \alphaf - \gamma_\graph 
\sqrt{\betaf \alphaf}}{1 - \gamma_\graph} \cdot \Delta n \; .
\]
We obtained a contradiction to the right-hand side of~(\ref{eq:edges-lwr-bnd}).

The contradiction follows from the assumption that there exists a heavy vertex in $U_A \cup U_B$,
and it is impossible to flip the directions of the edges such that the number of 
violations of the orientation property will decrease. 
We conclude that there is an $(\betaf/2, \alphaf/2)$-orientation in $\Edge$. 
\end{proof} 

Define the numbers $\theta_A$  and $\theta_B$ as follows. Let $\theta_A > 0$ ($\theta_B > 0$) be the largest number such that $\theta_A < \delta_A$ ($\theta_B < \delta_B$)
and $\quarter \theta_A \Delta$ ($\quarter \theta_B \Delta$, respectively) is integer.

The following theorem is the main result of this paper. 

\begin{theorem}
Let $\Code$ be defined as above, and assume that $\gamma_\graph \le \half \sqrt{\theta_A \theta_B}$. 
Then, the decoder in Figure~\ref{fig:lp-primal} is able to correct any error pattern 
of a size less than or equal to 
$$
\frac{\theta_A \theta_B - 2 \gamma_\graph \sqrt{\theta_A \theta_B}}{4 (1 - \gamma_\graph)} \cdot \Delta n
$$ 
in a codeword $\bldc \in \Code$. 
\label{thrm:more-errors-correct}
\end{theorem}

\begin{proof} 
Let $\Edge$ be the set of edges in error (for a received word $\bldy$), and assume 
that 
\[
|\Edge| \le \frac{\theta_A \theta_B - 2 \gamma_\graph \sqrt{\theta_A \theta_B}}{4 (1 - \gamma_\graph)} \cdot \Delta n \; . 
\]
Then, by Lemma~\ref{lemma:orientation}, there exists an $(\theta_A/4, \theta_B/4)$-orientation of $\Edge$.  

Therefore, 
we are able to construct a feasible solution for the dual LP problem, as follows.  
\begin{itemize}
\item
For the edges $e \notin \Edge$, we set the values of $\tau_{v,e}^{(\alpha)}$ in the same way as
we set the values of $\tau_{v,e}^{(\alpha)}$ for $e \notin E_1$
in the proof of Lemma~\ref{lemma:main-1}. 
\item
For the (directed) edge $(u,v) \in \Edge$, we set 
\[
\forall \alpha \in \ff \backslash \{c_e\} \; : \; 
\tau_{v,e}^{(\alpha)} = - \fhalf - \epsilon \; \mbox{ and } \; \tau_{u,e}^{(\alpha)} = \thalf \; , 
\]
and 
\[
\tau_{u,e}^{(c_e)} = \tau_{v,e}^{(c_e)} = \half \;.  
\]
\end{itemize}

These settings clearly satisfy all the constraints~(\ref{LP-polytope-1}) and~(\ref{LP-polytope-1-1}).  
Moreover, since for every $v \in A$ ($v \in B$) there are less than 
$\quarter \delta_A \Delta$ ($\quarter \delta_B \Delta$, respectively) incident edges $e \in \Edge$ with  
the corresponding $\tau_{v,e}^{(\alpha)} = -\fhalf - \epsilon$, using the same argument as in Lemma~\ref{lemma:main-1}, 
for $\epsilon$ small enough, we have that~(\ref{LP-polytope-2}) is also satisfied. 
\end{proof} 

\section{Discussion}
\label{sec:discussion}

The relative minimum distance of the code $\Code$ was shown in~\cite{Roth-Skachek} to satisfy~(\ref{eq:min-distance}). 
By taking a sufficiently large $\Delta$, this bound can be made arbitrarily 
close to $\delta_A \delta_B$. Thus, the analysis in Section~\ref{sec:constant-fraction} 
demonstrates that the decoder in Figure~\ref{fig:lp-primal} is able to correct any error 
pattern of size approximately $\ost$ of this lower bound. 
For comparison, the analysis in Section~\ref{sec:error-orientation} shows that the decoder  
is actually able to correct approximately four times more errors, than it was shown in Section~\ref{sec:constant-fraction}. Consequently, the
fraction of  correctable errors under the decoder in Figure~\ref{fig:lp-primal} is (at least) 
approximately $\quarter \delta_A \delta_B$. 

It is interesting to compare this result with other related works. Thus, in~\cite{Zemor01} 
the code $\Code$ with $\delta_A = \delta_B = \delta$ (for $0 < \delta < 1$) was considered, and a bit-flipping decoder 
was presented. This decoder corrects approximately $\quarter \cdot \delta^2$ fraction of errors. 
Similar result for binary codes was also obtained in~\cite{Feldman-capacity} by using a \emph{linear-programming} decoder 
and a slightly different definition of expander graph. 

However, the fraction of correctable errors in $\Code$ can be boosted close to $\half \delta_A \delta_B$ 
by using more advanced decoding techniques~\cite{Zemor03},~\cite{Roth-Skachek},~\cite{Skachek-Roth-2003}. 
It is still an open question whether the similar fraction of errors can be corrected by using decoder
based on linear-programming methods. 

The fraction of correctable errors grows with the size of the alphabet (as well as the relative minimum distance does). 
For example, consider a binary code $\Code$ having the same constituent code $\code = \code(v)$ for each $v \in V$.
If $\code$ is a random code of relative minimum distance $\delta$ and rate $r$, then we have (with high probability)
\[
        r \ge 1 - \entropy_2(\delta) - o(1) \; ,
\] 
where $\entropy_2(\cdot)$ is the binary entropy function. 
The rate of $\Code$ is at least $2r - 1$ and the fraction of the correctable errors is arbitrarily close to 
$\quarter \cdot \delta^2$. In Table~\ref{tab:binary}, we present the relations between the code rate and 
the lower bound on the fraction of correctable errors. 

\begin{table*}[htb]
	\centering
		\begin{tabular}{||l||c|c|c|c|c|c|c|c|c||}
			\hline
			\hline
			\mbox{Rate of $\Code$} & 0.1 & 0.2 & 0.3 & 0.4 & 0.5 & 0.6 & 0.7 & 0.8 & 0.9\\
			\hline
			\mbox{Fraction of correctable errors, $\times 10^{-4}$} & 22.14 & 15.76 & 10.82 & 7.086 & 4.346 & 2.422 & 1.160 & 0.4217 & 0.0786\\
			\hline
			\hline
		\end{tabular}
	\caption{Lower bound on the fraction of correctable errors for various rates of $\Code$, for binary alphabet.}
	\label{tab:binary}
\end{table*}

Next, consider a code $\Code$ over a large alphabet. Take $\code = \code(v)$ (for each $v \in V$) 
to be Generalized Reed-Solomon code of relative minimum distance $\delta$ and rate $r \ge 1 - \delta$.
In this case, we also have to require that $q \ge \Delta$. 
Table~\ref{tab:large} presents the relations between the rate of such $\Code$ and 
the fraction of correctable errors. 

\begin{table*}[htb]
	\centering
		\begin{tabular}{||l||c|c|c|c|c|c|c|c|c||}
			\hline
			\hline
			\mbox{Rate of $\Code$} & 0.1 & 0.2 & 0.3 & 0.4 & 0.5 & 0.6 & 0.7 & 0.8 & 0.9\\
			\hline
			\mbox{Fraction of correctable errors, $\times 10^{-2}$} & 5.0625 & 4.0 & 3.0625 & 2.250 & 1.5625 & 1.0 & 0.5625 & 0.250 & 0.0625\\
			\hline
			\hline
		\end{tabular}
	\caption{Lower bound on the fraction of correctable errors for various rates of $\Code$, for large alphabet.}
	\label{tab:large}
\end{table*}

\appendix

{\sl Sketch of the proof of Theorem~\ref{thrm:core-size}.}

Assume, by contrary, that $\graph$ contains a $(\zeta_A, \zeta_B)$-error core associated with $\bldy$. 
Let $E' \subseteq E$ be the set of edges in this error core, 
and $A' \subseteq A$ and $B' \subseteq B$ such that $A' \cup B'$ is the set of all the endpoints of the edges in $E'$. 
We have
\begin{itemize}
\item
for any $v \in A'$: $| \{ E(v) \cap E' \} | \ge \zeta_A \Delta$;
\item
for any $v \in B'$: $| \{ E(v) \cap E' \} |  \ge \zeta_B \Delta$. 
\end{itemize}

Consider a subgraph $\cH = (U_A \cup  U_B, \Edge)$ of $\graph$ with 
$U_A = A'$, $U_B = B'$ and $\Edge = E'$. Let $\alp = |U_A|/|A|$ and $\bet = |U_B|/|B|$. 
From Proposition~\ref{prop:degree}, we have
\begin{equation}
|E'| \le  ((1 - \gamma_\graph) \alp \bet 
+ \gamma_\graph \sqrt{\alp \bet}) \Delta n \; . 
\label{eq:error-core-bound}
\end{equation}

On the other hand, since $E'$ is the set of edges of an $(\zeta_A, \zeta_B)$-error core, 
we have 
\begin{equation}
|E'| \ge \alp n \cdot \zeta_A \Delta \quad \mbox{ and } \quad |E'| \ge \bet n \cdot \zeta_B \Delta \; . 
\label{eq:bound-E}
\end{equation}

There are two possibilities: 
\begin{description}
\item{Case 1:} $\alp \zeta_A \ge \bet \zeta_B$. Then, from~(\ref{eq:error-core-bound}) and~(\ref{eq:bound-E}), we have 
\[
\alp \zeta_A \le  ((1 - \gamma_\graph) \alp \bet + \gamma_\graph \sqrt{\alp \bet})\; , 
\]  
and so 
\[
\bet \ge \frac{\zeta_A - \gamma_\graph \sqrt{\bet/\alp}}{1 - \gamma_\graph} \ge \frac{\zeta_A - \gamma_\graph \sqrt{\zeta_A/\zeta_B}}{1 - \gamma_\graph} \; . 
\]
\item{Case 2:}
$\alp \zeta_A < \bet \zeta_B$. Then, from~(\ref{eq:error-core-bound}) and~(\ref{eq:bound-E}), similarly we have
\[
\alp \ge \frac{\zeta_B - \gamma_\graph \sqrt{\zeta_B/\zeta_A}}{1 - \gamma_\graph} \; . 
\]
\end{description}
In both cases, 
\[
|E'| \ge \frac{\zeta_A \zeta_B - \gamma_\graph \sqrt{\zeta_A \zeta_B}}{1 - \gamma_\graph} \cdot \Delta n \; ,
\]
in contradiction with the assumption. This concludes the proof.

\section*{Acknowledgments} 

The author wishes to thank Marcus Greferath.

%-------------------------------------------------------------------------------------

\begin{IEEEbiographynophoto}{Vitaly Skachek}
was born in Kharkov, Ukraine (former USSR), in 1973. He 
received the B.A. (Cum Laude), M.Sc. and Ph.D. degrees 
in computer science from the Technion---Israel Institute 
of Technology, in 1994, 1998 and 2007, respectively.  

During 1996--2002, he held various engineering positions. 
In the period 2002--2006, he has been working toward the Ph.D. 
degree at the Computer Science Department at the Technion.
In the summer of 2004, he visited the Mathematics of 
Communications Department at Bell Laboratories under
the DIMACS Special Focus Program in Computational 
Information Theory and Coding. During 2007--2009, Dr. Skachek was 
a postdoctoral fellow with the Claude Shannon Institute
and the School of Mathematical Sciences, University 
College Dublin. He is now a research fellow with 
the School of Physical and Mathematical Sciences,
Nanyang Technological University, Singapore. 

Dr.~Skachek is a recipient of the Permanent Excellent 
Faculty Instructor award, given by Technion. 
\end{IEEEbiographynophoto}

\end{document}